\documentclass[12pt,reqno]{amsart}
\pdfoutput=1 
\usepackage{amsmath,bbm}
\usepackage{latexsym}
\usepackage{amsfonts}
\usepackage{amssymb}
\usepackage{color}
\usepackage{graphicx}
\usepackage{url}
\usepackage{enumerate}
\usepackage{tikz}
\usetikzlibrary{shadings, intersections, calc, plotmarks}

\usepackage{geometry}
\xdefinecolor{tumblue}     {RGB}{0,101,189}
\xdefinecolor{tumgreen}    {RGB}{162,173,  0}
\xdefinecolor{tumred}      {RGB}{229, 52, 24}
\xdefinecolor{tumivory}    {RGB}{218,215,203}
\xdefinecolor{tumorange}   {RGB}{227,114, 34}
\xdefinecolor{tumlightblue}{RGB}{152,198,234}

\newtheorem{theorem}{Theorem}
\newtheorem*{theorem*}{Theorem}
\newtheorem{lemma}{Lemma}

\newtheorem*{corollary*}{Corollary}

\newcommand{\id}{{\rm{id}}} 

\newcommand{\cB}{{\mathcal{B}}}

\newcommand{\cH}{\mathcal{H}}

\newcommand{\cL}{\mathcal{L}}

\newcommand{\cT}{\mathcal{T}}

\newcommand{\1}{\mathbbm{1}}

\def\>{{\rangle}}
\def\<{{\langle}}
\newcommand{\be}{\begin{equation}}
\newcommand{\ee}{\end{equation}}
\newcommand{\bea}{\begin{eqnarray}}
\newcommand{\eea}{\end{eqnarray}}

\newcommand{\tr}[1]{\mathrm{tr}\left[#1\right]} 

\newcommand{\norm}[1]{\left\lVert #1 \right\rVert}

\begin{document}

\title[Quantum Zeno effect generalized]{Quantum Zeno effect generalized}

\author[M\"obus]{Tim M\"obus$^{1,2}$}
\author[Wolf]{Michael M. Wolf$^{1,3}$}
\email{m.wolf@tum.de}
\address{$^1$ Department of Mathematics, Technical University of Munich}
\address{$^2$ ETH Zurich, Switzerland}
\address{$^3$ Munich Center for Quantum
Science and Technology (MCQST),  M\"unchen, Germany}

\begin{abstract}The quantum Zeno effect, in its original form, uses frequent projective measurements to freeze the evolution of a quantum system that is initially governed by a fixed Hamiltonian. We generalize this effect simultaneously in three directions by allowing open system dynamics, time-dependent evolution equations and general quantum operations in place of projective measurements. More precisely, we study Markovian master equations with bounded generators whose time dependence is Lipschitz continuous. Under a spectral gap condition on the quantum operation, we show how frequent measurements again freeze the evolution outside an invariant subspace. Inside this space the evolution is described by a modified master equation.
\end{abstract}

\maketitle
\tableofcontents

\section{Introduction}\label{sec:intro}
The quantum Zeno effect predicts that frequent projective measurements can slow down and eventually freeze the evolution of a quantum system. It has been theoretically discovered and proven in \cite{Misra.1977} and its surprising consequences, in particular for decaying quantum systems, have been experimentally verified for instance in \cite{Itano.1990, Fischer.2001}. Mathematically, it is intimately related to the Trotter-Kato product formula \cite{Kato.1978}. Both can be proven by using one of the oldest tools in this context---Chernoff's $\sqrt{n}$-Lemma \cite{Chernoff.1968b}.

In this paper, we use Chernoff's Lemma together with standard tools from operator theory  to simultaneously generalize the quantum Zeno effect in three directions: we consider open quantum systems, allow time-dependent evolution equations, and use general quantum operations in place of projective measurements. More specifically, we consider an open quantum system whose time evolution is described by a master equation 
$$ \partial_t \rho(t)=\cL_t\big(\rho(t)\big),\qquad\rho(0)=\rho_0.$$
The generators $\cL_t$ are supposed to be bounded and with Lipschitz-continuous time-dependence. Assuming that this evolution gets repeatedly intercepted by a quantum operation that itself has a well-defined limit $P$ we show that the overall evolution becomes the solution of 
$$ \partial_t \tilde{\rho}(t)=\tilde{\cL}_t\big(\tilde{\rho}(t)\big),\quad\tilde{\rho}(0)=P(\rho_0)\quad\text{with }\tilde{\cL}_t=P\cL_t,$$
in the limit where the frequency of the interception goes to infinity.

The literature on the quantum Zeno effect is vast (see \cite{Schmidt.2003, Facchi.2008} for an overview) and we can only mention those works that appear to be closest to ours. For finite-dimensional quantum systems with time-independent evolution equations, the quantum Zeno effect has been generalized towards more general measurements for Hamiltonian dynamcis in \cite{Li.2013, Lidar.2013, Filippov.2017} and for Lindblad-type dynamics, in parallel to the present work, in \cite{Burgarth.2018a,Burgarth.2018b}. The work \cite{Burgarth.2018b} also allows for finite families of arbitrary quantum operations and in this way follow an interesting route that is not addressed in the present paper.

After completion of this work \cite{Moebus.2018}, we also became aware of the recent preprint \cite{Zimboras.2018} in which essentially the same questions have been addressed. The tools and perspectives slightly differ, but the results seem to be consistent.

\section{Preliminaries}\label{sec:prelim}
This section will introduce  basic concepts and fix the mathematical framework and notation. Throughout, $\cH$ will be a complex separable Hilbert space and $\cB(\cH)$ and $\cT(\cH)$ the spaces of bounded operators and trace-class operators on $\cH$, respectively. The  trace-norm, which makes $\cT(\cH)$ a Banach space, will be written $\norm{A}_1:=\tr{|A|}$ and the space of bounded linear maps $T\in\cB\big(\cT(\cH)\big)$ becomes itself a Banach space with  operator norm $\norm{T}:=\sup_{X\in\cT(\cH)\setminus 0}\norm{T(X)}_1/\norm{X}_1$. All convergence results in this work will be w.r.t.~this norm. In particular, all operator-valued integrals and derivatives are understood in the corresponding uniform operator topology.

We will describe the evolution of quantum systems in the Schr\"odinger picture. In this framework a  \emph{quantum channel} is a completely positive and trace-preserving map $T\in\cB\big(\cT(\cH)\big)$ describing the evolution of a density operator $\rho\mapsto T(\rho)$ in discrete time. We will write the identity channel as $\id:\rho\mapsto\rho$ and call a completely positive map a \emph{quantum operation} if it is trace-non-increasing. We will use the fact that   quantum channels and  quantum operations that have a non-trivial fixed-point space satisfy $\norm{T}= 1$ \cite{Russo.1966}  and $1\in{\rm spec}(T)\subseteq \mathbbm{D}_1$ \cite{Evans.1978} where $\mathbbm{D}_\delta:=\{z\in\mathbbm{C}\;|\;|z|\leq \delta\}$. 

Norm-continuous one-parameter semigroups  have bounded generators \cite{Engel.2000}. In the special case of a semigroup of quantum channels, also known as \emph{quantum dynamical semigroup}, we will denote such a bounded generator by $\cL$. It describes a continuous time evolution via
\begin{equation*}
\partial_t \rho(t)=\cL\big(\rho(t)\big),\qquad \rho(t)=e^{t\cL}\rho(0).
\end{equation*}

Our aim is to study the effect of a frequently repeated measurement or more  general  operation that intercepts an evolution like this in regular time intervals. To this end, the effect of the considered operation will  be described by a quantum operation $M:\cT(\cH)\rightarrow\cT(\cH)$. If $M$ is trace-preserving this means that no post-selection conditioned on the measurement outcome is made. In this case, the measurement outcome will not be used afterwards so that $M$ actually need not be a measurement but can as well be any quantum mechanical time evolution. A central assumption for all of our results is that $M$ has a discrete eigenvalue $1$ which is separated in magnitude from the rest of the spectrum. That is, we assume that 
\begin{equation}\label{eq:gapassumption}
\begin{gathered}
1\in{\rm spec}_d(M)\\
{\rm spec}(M)\subseteq\{1\}\cup\mathbbm{D}_{\delta}\quad \text{for some}\ \delta<1.
\end{gathered}
\end{equation}   
where ${\rm spec}_d$ denotes the discrete spectrum which consists of all isolated points of the spectrum with finite-dimensional Riesz projector (q.v.~Eq.(\ref{eq:spectralprojection})). Note that for a quantum operation $1\in {\rm spec}(M)$ already implies that $\norm{M}=1$, which in turn is necessary for obtaining a non-zero Zeno-limit later on.
Eq.(\ref{eq:gapassumption}) also implies that the curves $\Gamma:[0,2\pi]\ni\varphi\mapsto 1+ e^{i\varphi}(1-\delta)/2$ and $\gamma:[0,2\pi]\ni\varphi\mapsto (1+\delta)e^{i\varphi}/2$ both separate the two parts of the spectrum. Due to the upper semicontinuity of each separated part of the spectrum, there is an $\epsilon>0$ such that $\Gamma$ and $\gamma$ also separate the spectrum of $Me^{t\cL}$ for all $t\in [0,\epsilon]$ (cf.~\cite{Kato.1995} p.212).
So, the assumption in Eq.(\ref{eq:gapassumption})  enables us to define the Riesz projector
\begin{equation}\label{eq:spectralprojection}
P_t:=\frac{1}{2\pi i}\oint_\Gamma \Big(z\;\id - Me^{t\cL}\Big)^{-1} dz\quad\text{for all}\ t\in[0,\epsilon].
\end{equation}
The map $[0,\epsilon]\ni t\mapsto P_t$ is analytic (cf.~\cite{Kato.1995} p.368) and projection-valued (cf.~\cite{Kato.1995} p.178). In particular, $P:=P_0$ is the spectral projector corresponding to the eigenvalue $1$ of $M$ and since $P=\lim_{n\rightarrow\infty}M^n$, $P$ is a quantum operation in its own right.
\section{Time-independent case}\label{sec:t_indep}
\begin{theorem}[Quantum Zeno effect for open systems]\label{thm:1} Let $\cL$ be a bounded generator of a quantum dynamical semigroup on $\cT(\cH)$. If $M$ is a quantum operation on $\cT(\cH)$ that satisfies the spectral condition in Eq.(\ref{eq:gapassumption}) and whose spectral projector corresponding to the eigenvalue $1$ is denoted by P, then for all $t\geq 0$:
\begin{equation}
\norm{\left(M e^{t \frac{\cL}{n} }\right)^n - e^{t P\cL P}P}\rightarrow 0\qquad\text{for }n\rightarrow\infty.
\end{equation}
\end{theorem}
For finite-dimensional $\cH$ this result is also implied by  \cite{Burgarth.2018b} (Thm.1 therein).

We divide the proof of Thm.\ref{thm:1} into two parts, which are stated as Lem.\ref{lem:1} and Lem.\ref{lem:3} below. Thm.\ref{thm:1} is then obtained by combining the two Lemmas via the triangle inequality. In the reminder of this section, we will for ease of notation absorb the time parameter of $t\cL$ into $\cL$.
\begin{lemma}\label{lem:1} Under the assumptions of Thm.\ref{thm:1}  and with $P_t$ defined as in Eq.(\ref{eq:spectralprojection}): 
\begin{equation}
\norm{\left(M e^{\frac{\cL}{n} }\right)^n-\left(P_{\frac{1}{n}}M e^{\frac{\cL}{n}}P_{\frac{1}{n}}\right)^n}\rightarrow 0\qquad\text{for }n\rightarrow\infty.
\end{equation}
\end{lemma}
\begin{proof} We assume $n\geq 1/\epsilon$ and use the spectrum-separating curve $\gamma:[0,2\pi]\ni\varphi\mapsto (1+\delta)e^{i\varphi}/2$ for the holomorphic functional calculus (cf.~\S 2.3 in \cite{Simon.2015}). This allows us to write
\begin{eqnarray}
\norm{\left(M e^{\frac{\cL}{n} }\right)^n-\left(P_{\frac{1}{n}}M e^{\frac{\cL}{n}}P_{\frac{1}{n}}\right)^n} &=& \norm{\frac{1}{2\pi i}\oint_\gamma z^n \left(z\;\id - M e^{\frac{\cL}{n} }\right)^{-1} dz}\nonumber\\
&\leq& \left(\frac{1+\delta}{2}\right)^{n+1}\sup_{z,\tau} \norm{\left(z\;\id - M e^{\tau \cL }\right)^{-1}},\nonumber 
\end{eqnarray}
where the supremum is taken over the compact set $(z,\tau)\in\gamma\times [0,1/n]$ on which $(z\;\id - M e^{\tau \cL })$ is bijective. Since the inverse is continuous on the set of bijective bounded operators, the supremum is attained at a maximum, say $s(n)<\infty
$. The fact that $s(n)$ is non-increasing in $n$ together with $(1+\delta)/2<1$ then completes the proof.
\end{proof}

For the next step we need the following auxiliary Lemma by Chernoff \cite{Chernoff.1968b}:
\begin{lemma}[Chernoff's $\sqrt{n}$-Lemma]\label{lem:Chernoff} Let $X$ be a Banach space, $C\in\cB(X)$ a contraction and $\1$ the identity map on $X$. Then
\begin{equation}\label{eq:Chernoff}
\norm{C^n-e^{n(C-\1)}}\leq \sqrt{n}\norm{C-\1}\qquad \forall n\in\mathbbm{N},
\end{equation}
where $\norm{\cdot}$ is the operator norm on $\cB(X)$.
\end{lemma}
With this we can show the following remaining step in the proof of Thm.\ref{thm:1}:
\begin{lemma}\label{lem:3} Under the assumptions of Thm.\ref{thm:1}  and with $P_t$ defined as in Eq.(\ref{eq:spectralprojection}): \begin{equation}\label{eq:lem3}
\norm{\left(P_{\frac{1}{n}} M e^{\frac{\cL}{n}}P_{\frac{1}{n}}\right)^n -  e^{P\cL P}P}\rightarrow 0\qquad\text{for }n\rightarrow\infty.
\end{equation}
\end{lemma}
\begin{proof}
Again, we assume $n\geq 1/\epsilon$ so that the projection $P_\frac{1}{n}$ is well-defined. We aim at applying Chernoff's inequality from Lem.\ref{lem:Chernoff} with $X:=P_\frac{1}{n}\cT(\cH)$. Since $X$ is the kernel of $(\id-P_\frac{1}{n})$, it is a closed linear subspace of $\cT(\cH)$ and thus itself a Banach space when equipped with the norm $\norm{\cdot}_1$. Then $\1:=P_\frac{1}{n}$ is the identity in $\cB(X)$ and $C:=P_\frac{1}{n} Me^{\cL/n}P_\frac{1}{n}$ becomes a contraction in $\cB(X)$ to which we apply Lem.\ref{lem:Chernoff}. The central quantity to analyze then becomes
\begin{equation}
n(C-\1)= P_\frac{1}{n}M\cL P_\frac{1}{n}+n\left(P_\frac{1}{n} M P_\frac{1}{n}-P_\frac{1}{n}\right)+{\mathcal{O}}\left(\frac{1}{n}\right),\label{eq:nC1}
\end{equation} where we expanded the exponential $e^{\cL/n}$ to obtain the r.h.s. Before continuing with this expression, let us define $P':=\partial_t P_t\big|_{t=0}$ and note that exploiting $P_t=P_t^2$ before taking the derivative with the product rule implies $P'=PP'+P'P$. Additionally, the Riesz projector $P$ decomposes into the spectral projector w.r.t.~ the discrete eigenvalue $1$ and a nilpotent operator (cf.~\cite{Simon.2015} Cor.~2.3.6). The nilpotent part is zero because $M$ is a contraction (cf.~\cite{hasenoehrl.2020} p.57) so that $PM=P$. Therefore, expanding $ P_\frac{1}{n}=P+\frac{1}{n}P'+\mathcal{O}(1/n^2)$ yields
$$ n\left(P_\frac{1}{n} M P_\frac{1}{n}-P_\frac{1}{n}\right) = PMP'+P'MP-P' +\mathcal{O}(1/n)=\mathcal{O}(1/n).$$ 
Inserting this into Eq.(\ref{eq:nC1}), we obtain $n(C-\1)\rightarrow P\cL P$ for $n\rightarrow\infty$.
This in turn implies that the r.h.s. of Eq.(\ref{eq:Chernoff}) vanishes asymptotically so that Chernoff's Lemma finally leads to Eq.(\ref{eq:lem3}).
\end{proof}
Combining Lem.\ref{lem:1} with Lem.\ref{lem:3} via the triangle inequality then also completes the proof of Thm.\ref{thm:1}. 

Before extending the result to the time-dependent case, let us briefly discuss the generator $\tilde{\cL}:=P\cL P$ that appears in the theorem and that will also show up in the time-dependent case discussed in the next section. In general,  $\tilde{\cL}$ will not generate a quantum dynamical semigroup even if $P$ is a quantum channel, since complete positivity may be violated. However, this becomes cured when the evolution is restricted to $P\cT(\cH)$. On this restricted space, however, many generators that act differently on $\cT(\cH)$ lead to the same evolution. We may for instance choose $\tilde{\cL}\in\{P\cL P, P\cL,P(\id+\cL)-\id\}$ and they all lead to identical $e^{\tilde{\cL}}P$. Clearly, there are many more choices and depending on $\cL$ and $P$, these may or may not generate a completely positive evolution on $\cT(\cH)$.
For a particular example see \cite{Burgarth.2018a} (remark 4, example 1).

\section{Time-dependent case}\label{sec:t_dep}
In this section, we consider  time-dependent scenarios where the starting point is the solution of the time-dependent master equation
\begin{equation}\label{eq:ODEt}
\partial_t\rho(t)=\cL_t\big(\rho(t)\big),\qquad \rho(0)=\rho_0
\end{equation} in some finite time-interval $t\in[0,\tau]$. Throughout we will assume that $t\mapsto\cL_t\in\cB\big(\cT(\cH)\big)$ is an $L-$Lipschitz map into the set of bounded generators of quantum dynamical semigroups. The solution of Eq.(\ref{eq:ODEt}) can be expressed with the help of \emph{propagators} (a.k.a. evolution operators). These form a two-parameter family of quantum channels $T_{[t,s]}$ for $0\leq s\leq t\leq \tau$ so that $\rho(t)=T_{[t,s]}\big(\rho(s)\big)$. The propagators then satisfy (cf.~\cite{Cauchy.1984} Sec.7.1)
\begin{equation}\label{eq:T}
T_{[t,t]}=\id,\quad T_{[t,r]}T_{[r,s]}=T_{[t,s]},\quad {\partial_s} T_{[t,s]}=-T_{[t,s]}\cL_s,
\end{equation} for $t\geq r\geq s$ and we can regard the derivative in Eq.(\ref{eq:T})  in the uniform operator topology.

As before, we will intercept the evolution by applying a quantum operation $M$ in $n$ regular time-intervals and we are interested in the asymptotic behavior of the resulting evolution that is  described by the quantum operation 
\begin{equation}
T_n:=\prod_{i=1}^n \left( M T_{\big[\frac{i}{n}\tau,\frac{i-1}{n}\tau\big]}\right).
\end{equation} 
Here and in the following, products of this form have to be understood time-ordered, i.e., maps are composed in the order of the corresponding times. Our main result is the following natural extension of Thm.\ref{thm:1}:
\begin{theorem}[Zeno effect for time-dependent master equations]\label{thm:2}
Let $t\mapsto\cL_t\in\cB\big(\cT(\cH)\big)$ be a Lipschitz continuous map into the set of bounded generators of quantum dynamical semigroups and $M\in\cB\big(\cT(\cH)\big)$ be a quantum operation that satisfies the spectral condition in Eq.(\ref{eq:gapassumption}). Then the limit $\tilde{\rho}(\tau):=\lim_{n\rightarrow\infty}T_n(\rho_0)$ exists and coincides with the solution of 
\begin{equation}\label{eq:MEtilde}
\partial_t \tilde{\rho}(t)=P\cL_t P\big(\tilde{\rho}(t)\big),\quad \tilde{\rho}(0)=P(\rho_0),
\end{equation} 
where $P$ is the spectral projector corresponding to the eigenvalue $1$ of $M$.
\end{theorem}
Before coming to the proof of Thm.\ref{thm:2} we will establish the main Lemma that allows us to approximate the evolution that solves Eq.(\ref{eq:ODEt}) with piece-wise constant generators:
\begin{lemma}\label{lem4} If $t\mapsto\cL_t$ is $L$-Lipschitz and $T_{[t,s]}$ the family of evolution operators of the corresponding time-dependent master equation, then for all $t,s,\delta\geq 0$:
\begin{equation}\label{eq:Lem4}
\norm{T_{[t+\delta,t]}-e^{\delta \cL_s}}\leq L\left(\delta|t-s|+\frac{\delta^2}{2}\right).
\end{equation}
\end{lemma}
\begin{proof}
We divide the estimate into two parts given by the following triangle inequality:
\begin{equation}
\norm{T_{[t+\delta,t]}-e^{\delta \cL_s}}\leq \norm{T_{[t+\delta,t]}-e^{\delta \cL_t}}+\norm{e^{\delta \cL_t}-e^{\delta \cL_s}}\label{eq:triangle}.
\end{equation}
In order to bound the first term, we invoke the fundamental theorem of calculus, which together with Eq.(\ref{eq:T}) yields
\begin{eqnarray}
T_{[t+\delta,t]}-e^{\delta\cL_t}&=&-\int_t^{t+\delta} \partial_r\left(T_{[t+\delta,r]}e^{(r-t)\cL_t}\right) dr\nonumber\\
&=& \int_t^{t+\delta}T_{[t+\delta,r]}\big(\cL_r -\cL_t\big) e^{(r-t)\cL_t} dr.\nonumber
\end{eqnarray} 
Using that $\norm{T_{[t+\delta,r]}}=\norm{e^{(r-t)\cL_t}}=1$ and $\norm{\cL_r - \cL_t}\leq L|r-t|$ then proves that the first term on the r.h.s. of Eq.(\ref{eq:triangle}) is bounded by $L\delta^2/2$. In a similar vein, we can bound the second term by $\delta L|t-s|$ using
\begin{eqnarray}
e^{\delta\cL_t}-e^{\delta\cL_s}&=&\int_0^1 \partial_r\left( e^{(1-r)\delta\cL_s}e^{r\delta\cL_t}\right) dr \nonumber\\
&=&\delta\int_0^1 e^{(1-r)\delta\cL_s} \big(\cL_t-\cL_s\big)e^{r\delta\cL_t} dr \nonumber.
\end{eqnarray}
\end{proof}
\begin{proof}[Proof of Thm. 2] For ease of notation but w.l.o.g., we set $\tau=1$ so that $T_n$ describes an evolution for the time-interval $[0,1]$. The main idea of the proof is to use two levels of sub-divisions of this  time-interval: on the coarser level, we divide the interval into $m$ subintervals of roughly equal size within which we approximate the evolution governed by the initial master equation by using constant generators. On the finer level of $n\gg m$ partitions we then exploit Thm.\ref{thm:1} to obtain the claimed result up to an approximation error that vanishes in the limit of large $m$.  

More specifically, we partition $[n]:=\{1,\ldots, n\}$ into $m$ subsets with the help of a  step function $\theta:[n]\rightarrow[0,1], \theta(i):= \frac1m \lceil\frac{im}{n}\rceil$, whose range is $\frac1m [m]$.  For any $t\in [0,1]$ that is a multiple of $\frac1n$ we define the time-ordered products
\begin{equation}
W_{[1,t]}:=\prod_{i=1}^{n(1-t)} MT_{\big[t+\frac{i}{n},t+\frac{i-1}{n}\big]},\qquad W'_{[t,0]}:=\prod_{i=1}^{nt} M e^{\frac1n \cL_{\theta(i)}},\label{eq:Ws}
\end{equation}
with $W_{[1,1]}:=\id=:W'_{[0,0]}$. Using that these are quantum operations (and thus have norm bounded by one), we can exploit a telescopic sum together with Lem.\ref{lem4} to obtain
\begin{eqnarray}
\norm{W_{[1,0]}-W'_{[1,0]}} &=& \norm{ \sum_{i=0}^{n-1}W_{\big[1,\frac{i+1}{n}\big]}M\left(T_{\big[\frac{i+1}{n},\frac{i}{n}\big]}-e^{\frac1n \cL_{\theta(i+1)}}\right) W'_{\big[\frac{i}{n},0\big]}}\nonumber\\
&\leq &n \max_i \left\{\norm{T_{\big[\frac{i+1}{n},\frac{i}{n}\big]}-e^{\frac1n \cL_{\theta(i+1)}}}\right\}\nonumber\\
&\leq &L \max_i\left\{\left|\frac{i}{n}-\theta(i+1)\right|+\frac{1}{2n}\right\}\leq \frac{3L}{m},\label{eq:3Lm}
\end{eqnarray}
where the last step used $m\leq n$ and $|\theta(i+1)-i/n|\leq2/m$.

To emphasize the dependence on $n,m$, let us write $W'_{[1,0]}=:W'(n,m)$. Applying Thm.\ref{thm:1} we know that
$$ W'(m):=\lim_{n\rightarrow\infty} W'(n,m)=\prod_{j=1}^m e^{\frac1m P\cL_{j/m}P}P.$$
Moreover, $W':=\lim_{m\rightarrow\infty} W'(m)$ exists and corresponds to the solution of the master equation in Eq.(\ref{eq:MEtilde}).\footnote{In fact, this limit is the way in which the propagators are constructed in the first place (cf.~\cite{Tanabe.1997}, Sec.7.3).} This can be seen by essentially repeating the argument in Eq.(\ref{eq:Ws}) and Eq.(\ref{eq:3Lm}) with the corresponding generators and $M$ replaced by $\id$. Hence, by using $T_n=W_{[1,0]}$ we finally arrive at the desired result
\begin{eqnarray}
\limsup_{n\rightarrow\infty}\norm{T_n-W'} &\leq &\limsup_{m\rightarrow\infty}\limsup_{n\rightarrow\infty}\big(\norm{T_n-W'(n,m)}\nonumber\\ && \ +\norm{W'(n,m)-W'(m)}+\norm{W'(m)-W'}\big)=0\nonumber.
\end{eqnarray}
\end{proof}

\section{Discussion}\label{sec:discussion}
Finally, we want to discuss possible generalizations of the obtained results. 

The first and rather immediate observation is that the structure that is stated in the theorems is barely used in the proofs. Our interest was the realm of quantum theory, and therefore we decided to phrase the results in that language. However, it suffices mathematically to have an arbitrary Banach space (in place of $\cT(\cH)$), a contraction $M$ that satisfies the spectral condition in Eq.(\ref{eq:gapassumption}), and bounded generators $\cL_t$. The latter even need not generate a contraction semigroup as long as only finite times are considered.  In particular, `complete positivity' is not used in the proofs and for instance replacing it by `positivity' would be just as fine.

A second effortless generalization concerns the assumption that the time-intervals after which $M$ is applied are all equal. If we want to replace equally spaced intervals $[i/n,(i-1)/n]$ by varying intervals $[\nu(i/n),\nu((i-1)/n)]$ for some Lipschitz-function $\nu$, then this is equivalent to using equally spaced intervals with generators $\cL_{\nu(t)}$ in place of $\cL_t$. 

In fact, the two aforementioned points (general Banach spaces and non-uniform time-intervals) build the framework that is chosen in \cite{Zimboras.2018}.

Another rather straight forward generalization can be obtained by relaxing the Lipschitz assumption for $t\mapsto\cL_t$ to H\"older continuity for any H\"older exponent $\alpha\in (0,1]$. 

Finally, a possibly more desirable line of generalizations is to allow for unbounded generators. This direction is more challenging and beyond the scope of this paper. The reasoning in Sec.\ref{sec:t_dep} seems to be extendable to families of unbounded generators $\{\cL_t\}$ under suitable assumptions on the domain of the generators and the stability of the semigroups they generate (cf.~Chap.7 in \cite{Tanabe.1997}). The uniform operator topology then has to be left for the strong operator topology. However, identifying the natural conditions  under which an analogue of the ingredient from Sec.\ref{sec:t_indep} still holds true has to be left for future work. For special cases, existing results that point in this direction  in the time-independent scenario are summarized in \cite{Schmidt.2003}.

Finally, one may allow varying $M$ as it is initiated in \cite{Burgarth.2018b} or consider finite $n$, but these stories have to be told elsewhere.\vspace*{5pt}

\emph{Acknowledgments:} MMW acknowledges funding by the Deutsche Forschungsgemeinschaft (DFG, German Research Foundation) under Germany's Excellence Strategy – EXC-2111 – 390814868. Furthermore, we thank Simon Becker, Nilanjana Datta, and Robert Salzmann for pointing out the missing assumption in Eq.(\ref{eq:gapassumption}) (cf.~\cite{Becker.2020}).

\bibliographystyle{ieeetr}
\bibliography{Zeno}{}

\begin{thebibliography}{10}

\bibitem{Misra.1977}
B.~Misra and E.~C.~G. Sudarshan, ``{The Zeno’s paradox in quantum theory},''
  {\em Journal of Mathematical Physics}, vol.~18, no.~4, pp.~756--763, 1977.

\bibitem{Itano.1990}
W.~M. Itano, D.~J. Heinzen, J.~J. Bollinger, and D.~J. Wineland, ``{Quantum
  Zeno effect},'' {\em Phys. Rev. A}, vol.~41, pp.~2295--2300, 1990.

\bibitem{Fischer.2001}
M.~C. Fischer, B.~Guti\'errez-Medina, and M.~G. Raizen, ``{Observation of the
  Quantum Zeno and Anti-Zeno Effects in an Unstable System},'' {\em Phys. Rev.
  Lett.}, vol.~87, p.~040402, 2001.

\bibitem{Kato.1978}
T.~Kato and K.~Masuda, ``Trotter's product formula for nonlinear semigroups
  generated by the subdifferentials of convex functionals,'' {\em J. Math. Soc.
  Japan}, vol.~30, no.~1, pp.~169--178, 1978.

\bibitem{Chernoff.1968b}
P.~R. Chernoff, ``Note on product formulas for operator semigroups,'' {\em
  Journal of Functional Analysis}, vol.~2, no.~2, pp.~238--242, 1968.

\bibitem{Schmidt.2003}
A.~U. Schmidt, {\em Mathematical Physics Research at the Leading Edge},
  ch.~{Mathematics of the Quantum Zeno Effect}, pp.~111--141.
\newblock Nova Science Publishers, 2003.

\bibitem{Facchi.2008}
P.~Facchi and S.~Pascazio, ``{Quantum Zeno dynamics: mathematical and physical
  aspects},'' {\em Journal of Physics A: Mathematical and Theoretical},
  vol.~41, no.~49, p.~493001, 2008.

\bibitem{Li.2013}
Y.~Li, D.~A. Herrera-Mart\'{\i}, and L.~C. Kwek, ``{Quantum Zeno effect of
  general quantum operations},'' {\em Phys. Rev. A}, vol.~88, p.~042321, 2013.

\bibitem{Lidar.2013}
J.~M. Dominy, G.~A. Paz-Silva, A.~T. Rezakhani, and D.~A. Lidar, ``{Analysis of
  the quantum Zeno effect for quantum control and computation},'' {\em Journal
  of Physics A: Mathematical and Theoretical}, vol.~46, no.~7, p.~075306, 2013.

\bibitem{Filippov.2017}
I.~A. Luchnikov and S.~N. Filippov, ``Quantum evolution in the stroboscopic
  limit of repeated measurements,'' {\em Phys. Rev. A}, vol.~95, p.~022113,
  2017.

\bibitem{Burgarth.2018a}
D.~Burgarth, P.~Facchi, H.~Nakazato, S.~Pascazio, and K.~Yuasa, ``{Generalized
  Adiabatic Theorem and Strong Coupling Limits},'' {\em arXiv:1807.02036},
  2018.

\bibitem{Burgarth.2018b}
D.~Burgarth, P.~Facchi, H.~Nakazato, S.~Pascazio, and K.~Yuasa, ``{Quantum Zeno
  Dynamics from General Quantum Operations},'' {\em arXiv:1809.09570}, 2018.

\bibitem{Moebus.2018}
T.~M\"obus, ``{Generalized Quantum Zeno effects},'' {\em Bachelor's thesis,
  Department of Mathematics, Technical University of Munich}, July 19 2018.

\bibitem{Zimboras.2018}
N.~Barankai and Z.~Zimbor\'{a}s, ``Generalized quantum {Zeno} dynamics and
  ergodic means,'' {\em arXiv:1811.02509}, 2018.

\bibitem{Russo.1966}
B.~Russo and H.~A. Dye, ``{A note on unitary operators in $C\sp{\ast}$
  -algebras},'' {\em Duke Math. J.}, vol.~33, no.~2, pp.~413--416, 1966.

\bibitem{Evans.1978}
D.~E. Evans and R.~Hoegh-Krohn, ``{Spectral Properties of Positive Maps on
  C*-Algebras},'' {\em Journal of the London Mathematical Society}, vol.~s2-17,
  no.~2, pp.~345--355, 1978.

\bibitem{Engel.2000}
K.-J. Engel and R.~Nagel, {\em One-Parameter Semigroups for Linear Evolution
  Equations}, vol.~194 of {\em Graduate Texts in Mathematics}.
\newblock New York, NY: {Springer-Verlag New York Inc}, 2000.

\bibitem{Kato.1995}
T.~Kato, {\em Perturbation Theory for Linear Operators}, vol.~132.
\newblock Berlin and Heidelberg: Springer, 2~ed., 1995.

\bibitem{Simon.2015}
B.~Simon, {\em Operator theory: A comprehensive course in analysis, part 4}.
\newblock Providence, Rhode Island: {American Mathematical Society}, 2015.

\bibitem{hasenoehrl.2020}
M.~Haseno\"hrl and M.~M. Wolf, ``'interaction-free' channel discrimination,''
  {\em arXiv:2010.00623}, 2020.

\bibitem{Cauchy.1984}
H.~Fattorini and A.~Kerber, {\em The Cauchy Problem}.
\newblock Cambridge University Press, 1984.

\bibitem{Tanabe.1997}
H.~Tanabe, {\em Functional analytic methods for partial differential
  equations}, vol.~204 of {\em Monographs and textbooks in pure and applied
  mathematics}.
\newblock New York, NY: Dekker, 1997.

\bibitem{Becker.2020}
S.~Becker, N.~Datta, and R.~Salzmann, ``Quantum zeno effect for open quantum
  systems,'' {\em arXiv:2010.04121}, 2020.

\end{thebibliography}

\end{document}